\documentclass[prb,twocolumn,showpacs,amssymb,superscriptaddress,floatfix]{revtex4-2}
\usepackage{graphicx}
\usepackage{amsmath}
\usepackage{amsfonts}
\usepackage{amsthm}
\usepackage{amssymb}
\usepackage{amsbsy}
\usepackage{wasysym}
\usepackage{bm}
\usepackage{bbm}
\usepackage{mathrsfs}
\usepackage{color}
\usepackage{hyperref}
\usepackage{braket}
\usepackage{soul}
\usepackage{mathtools}
\usepackage{soul}

\hypersetup{
    colorlinks=true,
    linkcolor=red,        
    citecolor=blue,
    breaklinks=true,
    urlcolor=blue
}

\newtheorem{theorem}{Theorem}

\date{\today}
\begin{document}
\title{A theory of quasiballistic spin transport}

\author{Jeffrey Z. Song}

\affiliation{Department of Physics, University of Washington, Seattle, WA 98195, USA}

\author{Hyunsoo Ha}

\affiliation{Department of Physics, Princeton University, Princeton, New Jersey 08544, USA}

\author{Wen Wei Ho}
\affiliation{Department of Physics, National University of Singapore, Singapore 117551}
\affiliation{Centre for Quantum Technologies, National University of Singapore, 3 Science Drive 2, Singapore 117543}

\author{Vir B. Bulchandani}

\affiliation{Department of
Physics and Astronomy, Rice University, 6100 Main Street
Houston, TX 77005, USA}

\begin{abstract}
A recent work [Mierzejewski et al., Phys.~Rev.~B {\bf 107}, 045134 (2023)] observed ``quasiballistic spin transport'' -- long-lived and transiently ballistic modes of the magnetization density -- in numerical simulations of infinite-temperature XXZ chains with power-law exchange interactions. We develop an analytical theory of such quasiballistic spin transport. Previous work found that this effect was maximized along a specific locus in the space of model parameters, which interpolated smoothly between the integrable Haldane-Shastry and XX models and whose shape was estimated from numerics. We obtain an analytical estimate for the lifetime of the spin current and show that it has a unique maximum along a different locus, which interpolates more gradually between the two integrable points. We further rule out the existence of a conserved two-body operator that protects ballistic spin transport away from these integrable points by proving that a corresponding functional equation has no solutions. We discuss connections between our approach and an integrability-transport conjecture for spin.
\end{abstract}
\maketitle
\section{Introduction} 
Lattice models in physics generically exhibit diffusive transport of their local conserved charges at high temperature. The most commonly encountered exceptions to this rule arise in one spatial dimension, as a consequence of integrability. At one extreme, there exist models that can be mapped to non-interacting particles, such as the spin-$1/2$ XX~\cite{takahashi2005thermodynamics} and Haldane-Shastry~\cite{spinongas} chains, which therefore exhibit ballistic transport of all their local conserved charges~\cite{Antal,Sirker_2011,BulchandaniHa}. A more complicated intermediate behavior arises in the spin-$1/2$ XXZ chain, which can support ballistic, diffusive, or superdiffusive transport of spin, depending on the model anisotropy~\cite{Bertini}.

Thus, unlike in continuous systems, faster-than-diffusive transport in lattice models is generally linked to the presence of integrability~\cite{Castella_1995,Sirker_2011}. As one perturbs away from a given integrable point, such faster-than-diffusive dynamics usually gives way to normal diffusion. However, the relaxation timescales controlling this crossover to normal diffusion can be extremely long, even when the system under consideration is chaotic in all other respects~\cite{Ferreira_2020,DeNardis21,Roy_2023,McCarthy,McRoberts,wang2025breakdownsuperdiffusionperturbedquantum,Mierzejewski_2023,Yang_2024,HBS}. 

In particular, several recent works have observed numerical evidence for long-lived and transiently ballistic (or ``quasiballistic'') dynamics of conserved charges in systems that appear to be far from integrability. Throughout this paper, we follow Ref. \cite{Mierzejewski_2023} in referring to this phenomenon as ``quasiballistic transport'', with the understanding that this terminology refers to transiently ballistic dynamics rather than the specific behavior of any conventional linear-response transport coefficient. For example, spin transport in power-law-interacting spin-$1/2$ XXZ chains~\cite{Mierzejewski_2023,Yang_2024} (see also \cite{Ferreira_2020}) and charge transport in interacting fermionic lattice models in any spatial dimension~\cite{HBS} can exhibit this phenomenon. Both sets of work focus on infinite temperature, where such behavior is intuitively least likely to occur. In the fermionic case, this physics was explained theoretically by computing an effective lifetime for the charge current, which could be made arbitrarily small as a function of the interaction range. Meanwhile, no detailed theoretical explanation has been proposed for quasiballistic spin transport in general and providing such an explanation is the primary goal of this work. 

The specific findings to be explained are as follows. It has been argued~\cite{Mierzejewski_2023,Yang_2024} that for spin-$1/2$ XXZ chains with power-law exchange interactions $J(r) \propto 1/r^{\alpha}$ between spins, long-lived and transiently ballistic modes of spin can arise for all values of the power-law exponent $\alpha$, with maximal enhancement of these ballistic modes for model anisotropies $\Delta^*(\alpha) \approx e^{2-\alpha}$ interpolating smoothly between the Haldane-Shastry chain (at $\alpha=2$) and the XX chain (as $\alpha \to \infty$), both of which exhibit ballistic spin transport as discussed above. However, the evidence that this locus exists and the extrapolation of its shape have so far been purely numerical. Our main contribution in this paper is an analytical prediction for the locus $\Delta^*(\alpha)$ on which the lifetime of the spin current is maximized, as a ratio of two infinite series. A nontrivial consistency check on our prediction is that it interpolates smoothly between the two ballistic points. However, our expression has the asymptotic behavior $\Delta^*(\alpha) \sim 3 \cdot 2^{-\alpha}$ as $\alpha \to \infty$, whose exponential dependence on $\alpha$ matches the earlier conjecture from numerics~\cite{Mierzejewski_2023,Yang_2024}, but decays more slowly when $\alpha$ is large.

We emphasize that Hamiltonians of this type are not merely a theoretical curiosity. Recent advances in quantum devices have renewed interest in the dynamics of long-range interacting spin chains~\cite{Bohnet_2016,Schuckert20,Zu_2021,Joshi_2022,franke2023quantum}. Importantly, these platforms inherently favour the study of real-time dynamics. Moreover, their relatively high effective temperatures, which can be undesirable when it comes to exploring traditional questions about ground-state physics, are ideal for studying the far-from-equilibrium transport properties of interest in this work. There now exist several concrete proposals~\cite{arrazola2016digital,Bermudez17,Birnkammer22} for how to realize precisely the power-law interacting XXZ Hamiltonians studied here using trapped ions and a recent experiment~\cite{Kranzl23} probing their dynamics. It should, however, be borne in mind that relatively small power-law exponents $0\leq \alpha \leq 3$ are most physically natural in these platforms~\cite{fossfeig2024progresstrappedionquantumsimulation}, while the discussion in our paper applies to $\alpha > 3/2$.

Before outlining our analytical approach, it is instructive to consider what a complete theory of quasiballistic spin transport would entail. First, one would have to rule out the possibility of ballistic spin transport for $\alpha \neq 2$ in the thermodynamic limit; while this seems highly unlikely, it is not possible to rule out using any theoretical technique that we are aware of. Assuming that ballistic transport could indeed be ruled out, the most convincing estimate of the relaxation time of a transiently ballistic spin current would then come from the spin diffusion constant, whose rigorous evaluation from the Kubo formula similarly appears to lie beyond the reach of existing analytical techniques.

In this work, we pursue a less ambitious yet analytically tractable version of the path sketched above. Rather than ruling out ballistic spin transport fully for $\alpha \neq 2$, we rule out the possibility of ballistic spin transport protected by a strictly two-body conservation law, as it is at the Haldane-Shastry~\cite{haldane1994physicsidealsemiongas,Bernevig_2001,Sirker_2011} and XX points. Similarly, rather than attempting to evaluate the spin diffusion constant from the Kubo formula, we estimate the lifetime of the spin current directly from its instantaneous rate of change, a strategy that has proved effective in fermionic systems~\cite{HBS} (see also~\cite{Kim_2015}) and is reasonable in the absence of ballistic transport. We conclude by discussing the scope of our approach and the connection between integrability and transport more broadly. In particular, we prove that the only spin-$1/2$ chains of the form Eq. \eqref{eq:Ham} with an exactly conserved spin current are Haldane-Shastry and XX chains (Theorem \ref{thm:thm1} of Appendix \ref{appendix:perfball}).

\section{Decay rate of the spin current}
We consider infinite, translation-invariant, inversion-symmetric spin-$1/2$ chains with global spin-rotation symmetry about the $z$-axis. In the absence of external magnetic fields, the most general two-body Hamiltonian with these properties is given by
\begin{align}
    \hat{H} = \sum_{\substack{m,n\in\mathbb{Z} \\ m < n}}a(m-n)(\hat{S}^x_m\hat{S}^x_{n} + \hat{S}^y_{m}\hat{S}^y_{n}) + b(m-n)\hat{S}^z_m\hat{S}^z_{n},\label{eq:Ham}
\end{align}
where $a(r)$ and $b(r)$ are arbitrary even functions and it will be convenient~\footnote{For power-law exponents $1/2 < \alpha \leq 3/2$, there is evidence~\cite{Schuckert20,Joshi_2022} that transport in models of this type crosses over to superdiffusion, which in turn gives way to unphysical instantaneous relaxation for $\alpha < 1/2$. Indeed, the analysis in our paper is complicated for $\alpha \leq 3/2$ by the failure of Eq. \eqref{eq:currentnorm} to converge. Restoring a finite-size cutoff in Eq. \eqref{eq:lifetime} and performing a numerical finite-size scaling analysis, not shown, we see behavior consistent with the qualitative picture proposed in Ref.~\cite{Schuckert20} and therefore with the absence of quasiballistic transport for $\alpha \leq 3/2$.} to assume that $a(r),\,b(r) = o(r^{-3/2})$ as $r \to \infty$. We additionally assume that $a$ is non-zero (i.e.~$a(r) \neq 0$ for some integer $r \neq 0$) to exclude Ising models without spin transport.

The magnetization is a local conserved charge of the Hamiltonian Eq. \eqref{eq:Ham} because the Heisenberg-picture spin operator $\hat{S}_n^z(t)$ at each site satisfies the discrete, operator-valued continuity equation
\begin{equation}
    \partial_t \hat{S}_n^z+\hat{j}_{n+1}-\hat{j}_n=0,
\end{equation}
where the local spin-current operator can be defined as
\begin{equation}
        \hat{j}_n = \sum_{r=1}^\infty a(r)\sum_{m=n}^{n+r-1} \hat{S}_{m-r}^{x}\hat{S}_{m}^{y}-\hat{S}_{m}^{x}\hat{S}_{m-r}^{y}.
\end{equation}
By the continuity equation, the total current operator coincides with the rate of change of the ``many-body position operator''~\cite{lieb1961two}, namely $\hat{J}=\sum_{n \in \mathbb{Z}} \hat{j}_n = \sum_{n \in \mathbb{Z}}  n \partial_t \hat{S}^z_n$, where 
\begin{equation}
\label{eq:spincurrent}
\hat{J} = \sum_{n\in\mathbb{Z}} \sum_{r=1}^{\infty} r a(r) (\hat{S}_{n}^x \hat{S}_{n+r}^y - \hat{S}_{n}^y \hat{S}_{n+r}^x).
\end{equation}
For generic, chaotic, inversion-symmetric Hamiltonians we expect diffusive spin transport above zero temperature (see Appendix \ref{appendix:twobody} for details). In such systems, a reasonable estimate for the infinite-temperature decay rate of the spin current can be obtained as follows~\cite{HBS,Kim_2015}. 

We first note that the rate of change $\dot{\hat{J}} = i [\hat{H},\hat{J}]$ of the spin current is given by the three-body operator  
\begin{align}
    \dot{\hat{J}} = \sum_{n\in\mathbb{Z}} \sum_{\substack{u,v \in \mathbb{Z}_{\neq 0} \\ u+v > 0}} D(u,v)(\hat{S}^x_{n+u} \hat{S}^x_{n-v} + \hat{S}^y_{n+u} \hat{S}^y_{n-v}) \hat{S}^z_n,
    \label{eq:dJ/dt}
\end{align}
where the function $D(u,v) = (u-v) a(u) a(v) + (u+v) a(u+v) (b(u)-b(v))$. At infinite temperature, this operator has zero mean but a non-zero, extensive variance, with a density
\begin{equation}
   \lim_{L\to \infty} \frac{\langle \dot{\hat{J}}(t)^2 \rangle_{\beta=0}}{L} =\frac{1}{64}\sum_{\substack{u,v\in \mathbb{Z}_{\neq 0} \\ u+v \neq 0}} D(u,v)^2
\end{equation}
in the thermodynamic limit (for concreteness, assume that the model Eq. \eqref{eq:Ham} is truncated to finite systems of $L$ sites by discarding interaction ranges $r > L$). In order to yield a characteristic timescale for the decay of the spin current, this must be normalized by the variance of latter, whose limiting density
\begin{equation}
\label{eq:currentnorm}
\lim_{L\to \infty} \frac{\langle \hat{J}(t)^2 \rangle_{\beta=0}}{L}  = \frac{1}{16} \sum_{n \in \mathbb{Z}_{\neq 0}} n^2a(n)^2
\end{equation}
is both non-zero and finite by assumption. Since both these quantities are independent of time, their ratio yields a dimensionful estimate for the instantaneous decay rate of $\hat{J}(t)$ for all time, namely
\begin{equation}\tau_{\mathrm{eff}}^{-1}=\lim_{L \to \infty}  \sqrt{\langle \dot{\hat{J}}(t)^2 \rangle_{\beta=0}/\langle \hat{J}(t)^2 \rangle_{\beta=0}}. 
\label{eq:lifetime}
\end{equation}
In particular, minimizing $\tau_{\mathrm{eff}}^{-1}$ over the space of model parameters $a$ and $b$ should maximize the lifetime of the spin current and therefore lead to a maximal enhancement of quasiballistic spin transport~\cite{HBS}. We note that this strategy is reasonable
provided $\hat{J}$ has negligible overlap with conserved quantities in the thermodynamic limit~\cite{Zotos_1997,Sirker_2011,ProsenBound}. We will revisit this assumption more carefully below. First, we verify that our formula Eq. \eqref{eq:lifetime} is sufficient to capture the qualitative features of quasiballistic spin transport that have been observed numerically in long-range interacting XXZ chains~\cite{Mierzejewski_2023,Yang_2024}. Our approach is similar in spirit to the ``moment method'' for approximating the dynamical spin structure factor~\cite{de1958inelastic,redfield1968moment}.

\section{Application to long-range interacting XXZ chains}
\subsection{Existence and uniqueness of an optimal anisotropy}
Let us now specialize to spin chains with XXZ-type anisotropy, $b(r) = \Delta a(r)$. Then we can ask the following question: for a given choice of exchange interactions $a(r)$, is there a value $\Delta = \Delta^*$ of the model anisotropy that maximizes the lifetime of the spin current, as measured by $\tau_{\mathrm{eff}}$, and is this maximum unique? Existence and uniqueness of such a maximum is perhaps the simplest qualitative feature of the previous numerical studies~\cite{Mierzejewski_2023,Yang_2024}, and our first contribution is to show that this follows from our Eq. \eqref{eq:lifetime} on general grounds.

Specifically, for XXZ-type interactions we find that
\begin{equation}
\label{eq:tauXXZ}
\tau_{\mathrm{eff}}^{-1}(\Delta) = \frac{1}{2}\sqrt{\frac{A \Delta^2 + 2 B \Delta + C}{\sum_{n \in \mathbb{Z}_{\neq 0}} n^2 a(n)^2}},
\end{equation}
where
\begin{align}
\nonumber 
A &= \sum_{\substack{u,v \in \mathbb{Z}_{\neq 0}\\ u+v \neq 0}} (u+v)^2 a(u+v)^2 (a(u)-a(v))^2, \\ 
\nonumber 
B &= \sum_{\substack{u,v \in \mathbb{Z}_{\neq 0}\\ u+v \neq 0}} (u^2-v^2)a(u)a(v)a(u+v)(a(u)-a(v)), \\
C &= \sum_{\substack{u,v \in \mathbb{Z}_{\neq 0}\\ u+v \neq 0}} (u-v)^2 a(u)^2 a(v)^2.
\label{eq:defAB}
\end{align}
In particular, since $A > 0$ by assumption, it follows by Eq. \eqref{eq:tauXXZ} that $\tau_{\mathrm{eff}}^{-1}$ has a unique, global minimum at 
\begin{equation}
\label{eq:Deltastar}
\Delta^* = -B/A,
\end{equation}
which was to be explained.

\begin{figure}[!t]
\includegraphics[width=\linewidth]{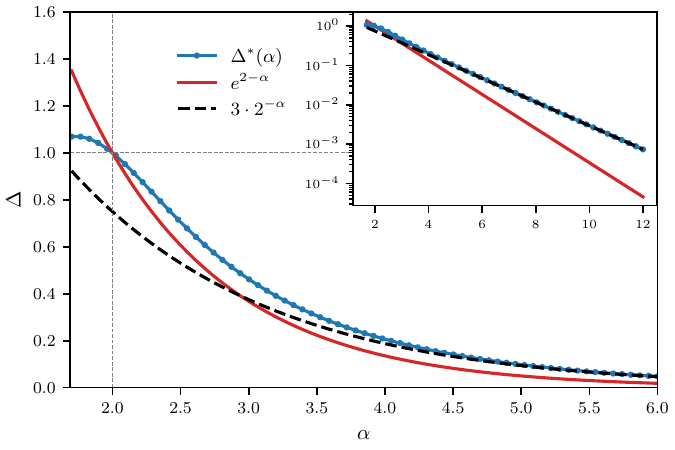}
\caption{\textbf{Main figure}: predicted optimal anisotropy $\Delta^*(\alpha)$ for quasiballistic spin transport as a function of the power-law exponent $\alpha$ (blue line). This is compared to the previous estimate $\Delta_{\mathrm{num}}^*(\alpha) \approx e^{2-\alpha}$ from numerics~\cite{Mierzejewski_2023, Yang_2024} (red line). Both these loci intersect the integrable Haldane-Shastry point when $\alpha = 2$ and recover the integrable XX point as $\alpha \to \infty$. \textbf{Inset}: tails of $\Delta^*(\alpha)$ versus $e^{2-\alpha}$ for large $\alpha$. The distinct asymptotic approach to the integrable XX point between the two estimates is clear.}
\label{fig:delta}
\end{figure}

\subsection{Power-law interactions}
We now specialize further to power-law exchange interactions $a(x)=1/|x|^\alpha$, since these are both experimentally motivated~\cite{arrazola2016digital,Bermudez17,Birnkammer22,Kranzl23} and the subject of the previous numerical studies~\cite{Mierzejewski_2023,Yang_2024} of quasiballistic spin transport. We further assume that $\alpha > 3/2$ to ensure consistency of our approach. Then Eq. \eqref{eq:Deltastar} predicts an exact expression for the optimal $\Delta^*$ for quasiballistic behavior as a function of $\alpha$. The resulting locus in the space of model parameters is a ratio of two infinite series, namely
\begin{equation}
\label{eq:powerlawdeltapred}
    \Delta^*(\alpha) = \frac{\sum_{\substack{u,v \in \mathbb{Z}_{\neq 0}\\ u+v \neq 0}} \frac{(u^2-v^2)(|u|^{\alpha}-|v|^\alpha)}{|u|^{2\alpha}|v|^{2\alpha}|u+v|^\alpha}}{\sum_{\substack{u,v \in \mathbb{Z}_{\neq 0}\\ u+v \neq 0} } \frac{(|u|^\alpha - |v|^\alpha)^2}{|u|^{2\alpha}|v|^{2\alpha}|u+v|^{2(\alpha-1)}}}.
\end{equation}
This curve is plotted in Fig. \ref{fig:delta}, where it is compared to the previous estimate~\cite{Mierzejewski_2023,Yang_2024} 
$\Delta^*_{\mathrm{num}}(\alpha) \approx e^{2-\alpha}$ for optimizing quasiballistic transport based on numerics. A nontrivial check on our prediction Eq. \eqref{eq:powerlawdeltapred} is that just like the numerical estimate $\Delta^*_{\mathrm{num}}(\alpha)$, it ``knows about'' integrability. Specifically, Eq. \eqref{eq:powerlawdeltapred} recovers the integrable Haldane-Shastry model, $\Delta^*(2)= 1$. Similarly, Eq. \eqref{eq:powerlawdeltapred} recovers the integrable XX model in the limit of nearest-neighbour interactions, $\Delta^*(\alpha) \to 0$ as $\alpha \to \infty$. Thus our analytical prediction recovers the two known integrable points in the space of power-law interacting XXZ models. (Note that strictly speaking, $\alpha=\infty$ is an integrable line.)

However, it is clear from Fig. \ref{fig:delta} that the approach of Eq. \eqref{eq:powerlawdeltapred} to the nearest-neighbour XX point is distinct from that of $\Delta_{\mathrm{num}}^*(\alpha)$, and in fact we find a slower approach
\begin{equation}
\label{eq:Deltastarlargealpha}
\Delta^*(\alpha) \sim 3 \cdot 2^{-\alpha}
\end{equation}
as $\alpha \to \infty$ than predicted before. To see this, note that for large $\alpha$, both $A$ and $B$ in Eq. \eqref{eq:defAB} are dominated by the contributions of the four lowest order terms  $(u,v) \in \{(1,-2),(-1,2),(2,-1),(-2,1)\}$. In particular, we find that $A \sim 4$ and $B\sim -12 \cdot 2^{-\alpha}$ as $\alpha \to \infty$, and substituting into Eq. \eqref{eq:Deltastar} recovers Eq. \eqref{eq:Deltastarlargealpha}.

\subsection{Optimal relaxation rate and perturbative regime}

\begin{figure}[t]
\centering
\includegraphics[trim=5.4pt 9pt 2pt 7pt, clip, width=\linewidth]{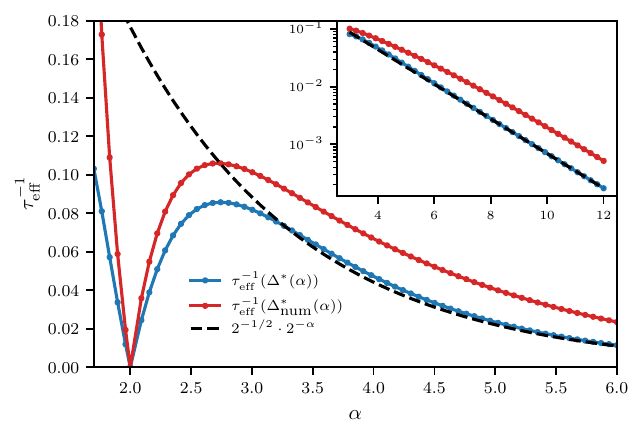}
\caption{\textbf{Main figure}: minimum decay rate for the spin current $\tau_{\mathrm{eff}}^{-1}(\Delta^*(\alpha))$ predicted by our approach (blue line), versus the decay rate $\tau_{\mathrm{eff}}^{-1}(\Delta_{\textrm{num}}^*(\alpha))$ that we estimate from the previous numerical extrapolation of the optimal anisotropy (red line). The decay rate vanishes as expected at the Haldane-Shastry point $\alpha =2$ and approaching the XX point as $\alpha \to \infty$. \textbf{Inset}: tails of  $\tau_{\mathrm{eff}}^{-1}(\Delta^*(\alpha))$ and $\tau_{\mathrm{eff}}^{-1}(\Delta_{\textrm{num}}^*(\alpha))$ for large $\alpha$. It is apparent that both estimates decay exponentially in $\alpha$ at the same rate, with our value of $\Delta^*(\alpha)$ yielding a constant-factor ($\approx 1/3$) reduction in the attainable decay rate as $\alpha \to \infty$.}
\label{fig:relaxation}
\end{figure}

Finally, it is instructive to consider the optimal relaxation rate predicted by our approach. This is given by combining Eqs. \eqref{eq:tauXXZ} and \eqref{eq:Deltastar}  to  obtain $\tau_{\mathrm{eff}}^{-1}(\Delta^*(\alpha))$, which is plotted in Fig. \ref{fig:relaxation}. As expected, this vanishes at $\alpha = 2$ and as $\alpha \to \infty$. We further find that the relaxation rate remains small over the entire interval $2 < \alpha < \infty$, consistent with the previous observation that the power-law XXZ chain is never ``far'' from ballistic transport~\cite{Mierzejewski_2023}.

While this is perhaps surprising for intermediate values of $\alpha$, it is expected that as $\alpha \to \infty$, power-law exchange interactions will merely yield a small (integrability-breaking) next-nearest-neighbour ``hopping'' correction $a(2) = 2^{-\alpha}$ to the integrable XX point~\cite{Mierzejewski_2023}. Moreover, for anisotropies $\Delta$ that are exponentially small in $\alpha$, the nearest-neighbour Ising ``interaction'' term $b(1) = \Delta$ will also be a small (integrability-preserving) perturbation to the XX point. This raises the questions of which small perturbation bottlenecks (i.e. lower bounds) our estimate of the relaxation rate of the spin current for $\alpha \gg 1$ and very small values of $\Delta$
-- on physical grounds, this should be the integrability-breaking hopping term -- and how far our estimate for the optimal lifetime $\tau_{\mathrm{eff}}(\Delta^*(\alpha))$ improves upon the lifetime $\tau_{\mathrm{eff}}(\Delta^*_{\mathrm{num}}(\alpha))$ that we estimate from the previous numerical extrapolation~\cite{Mierzejewski_2023,Yang_2024} in this perturbative regime. (We note that the emergence of quasiballistic spin transport from similar perturbations to the XX model was studied previously using different techniques~\cite{Ferreira_2020}.)

To explore these questions quantitatively, suppose that $\Delta \sim k \cdot c^{-\alpha}$ as $\alpha \to \infty$ for some constant $k$ and base $c > 1$. Subject to this assumption, we find that the dominant contribution to the decay rate $\tau_{\mathrm{eff}}^{-1}$ depends sensitively on $c$, with
\begin{equation}
\tau_{\mathrm{eff}}^{-1} \sim  \begin{cases} \frac{\Delta}{\sqrt{2}}, & c < 2, \\ \sqrt{\frac{\Delta^2 -6\cdot 2^{-\alpha}\Delta + 10\cdot 4^{-\alpha}}{2}}, & c=2, \\
\sqrt{5} \cdot 2^{-\alpha}, & c > 2,
\end{cases}
\end{equation}
as $\alpha \to \infty$. For suboptimal choices of $\Delta$ with $c \neq 2$, these formulas imply that $\tau_{\mathrm{eff}}^{-1}$ is lower-bounded by the contribution $\sim \sqrt{5}\cdot 2^{-\alpha}$ from integrability-breaking next-nearest-neighbour hopping, as expected. However, for the specific value $c=2$, a nontrivial cancellation between hopping and interaction terms in Eq. \eqref{eq:tauXXZ} predicts a potentially slower relaxation rate $\tau^{-1}_{\mathrm{eff}} \sim \sqrt{\frac{(k-3)^2+1}{2}} \cdot 2^{-\alpha}$, which attains its minimum when $k=3$. It follows that $\tau^{-1}_{\mathrm{eff}}(\Delta^*_{\mathrm{num}}(\alpha)) \sim \sqrt{5} \cdot 2^{-\alpha}$ is dominated by the hopping contribution, while $\tau^{-1}_{\mathrm{eff}}(\Delta^*(\alpha)) \sim (1/\sqrt{2})\cdot 2^{-\alpha}$ by Eq. \eqref{eq:Deltastarlargealpha}, predicting a decay rate for the spin current that is more than three times slower than would arise from next-nearest-neighbour hopping alone, see Fig. \ref{fig:relaxation}.

\section{Discussion}
We have developed a theory of quasiballistic spin transport, showing that the analytical approach proposed for charge transport in Ref.~\cite{HBS} can be extended to spin transport and qualitatively explains the numerical findings of Refs.~\cite{Mierzejewski_2023,Yang_2024}. Our analysis of this phenomenon in the perturbative regime of large $\alpha$ further suggests that the observed quasiballistic behavior is enabled by a non-trivial interplay between hopping and interaction terms that suppresses the decay of the spin current more strongly than one might na{\"i}vely expect.

This analytical approach is reasonable for the power-law interacting XXZ models in this paper provided they are not ``secretly'' ballistic away from the known integrable points $\alpha = 2$ and $\alpha \to \infty$. More precisely, our approach rules out perfectly ballistic spin transport (i.e. exact conservation of the spin current operator) away from these points, but it does not rule out imperfectly ballistic spin transport, as exemplified by the nearest-neighbour spin-$1/2$ XXZ chain with anisotropy $0 < |\Delta| < 1$. In the latter case, the spin current is not exactly conserved, but it has sufficiently large overlap with a conserved operator in the thermodynamic limit that the Mazur bound~\cite{mazur1969non,suzuki1971ergodicity} guarantees ballistic spin transport~\cite{Zotos_1997,ProsenBound} and the approach taken in this paper becomes invalid.

While we cannot exclude this possibility based on available theoretical techniques, we prove a partial result justifying our approach in Appendix \ref{appendix:twobody}. Specifically, we show that for the models studied in this paper, there is no two-body conserved charge that could give rise to a non-zero Mazur bound for $\Delta \neq 0, \pm 1$. Unfortunately, this does not eliminate the possibility of quasilocal, many-body conserved charges that protect the spin current, as in the nearest-neighbour XXZ chain~\cite{ProsenBound,Ilievski_2016}. At the same time, a systematic construction of quasilocal or pseudolocal~\cite{Ilievski_2016} conserved charges for an arbitrary lattice model without prior knowledge of its underlying integrable structure is an extremely challenging problem whose solution would have immediate implications for various important open problems in many-body physics, such as the rigorous justification of many-body localization~\cite{Basko_2006,OganesyanHuse,VoskAltman,deroeck2023rigoroussimpleresultsslow}.

Meanwhile, there is a long-standing expectation in the literature that ballistic transport of \textit{some} local conserved charge should always follow from integrability; this belief is sometimes referred to as the ``integrability-transport conjecture''~\cite{Castella_1995,Zotos_1997,Sirker_2011}. While ballistic energy transport in the nearest-neighbour XXZ model is consistent with this conjecture, the precise relationship between integrability and spin transport is murkier~\cite{Bertini}. We propose the following revision of the integrability-transport conjecture for inversion-symmetric lattices of spins: spin transport is ballistic at all temperatures only if the system is integrable. We do not know of any counterexamples to this conjecture, but it seems difficult to prove for the reasons discussed above.

A corollary of our analysis is a simplified integrability-transport theorem (see Theorem \ref{thm:thm1} of Appendix \ref{appendix:perfball}): for spin-$1/2$ chains of the form Eq. \eqref{eq:Ham}, we show that spin transport is only perfectly ballistic if the model is equivalent to either a nearest-neighbour XX chain or a Haldane-Shastry chain. Extending this result to imperfect ballistic transport and pseudolocal conserved charges, along the lines of recent rigorous results on non-integrability of certain spin chains~\cite{yamaguchi2024completeclassificationintegrabilitynonintegrability,yamaguchi2024proofabsencelocalconserved,shiraishi2025complete}, seems desirable for understanding the relationship between integrability and transport more generally.

\section{Acknowledgments}
We thank D.~A. Huse for numerous helpful discussions and collaborations on related work and M.~Mierzejewski for drawing quasiballistic spin transport to our attention and comments on the manuscript. H.~H.~is supported by NSF QLCI grant OMA-2120757. W.~W.~H.~is supported by the National Research Foundation (NRF), Singapore, through the NRF Felllowship NRF-NRFF15-2023-0008, and through the National Quantum Office, hosted in A*STAR, under its Centre for Quantum Technologies Funding Initiative (S24Q2d0009). 
\bibliography{bibl.bib}

\appendix
\onecolumngrid

\section{Sufficient conditions for ballistic spin transport}
In this Appendix, we discuss various sufficient conditions for the spin-chain Hamiltonian Eq. \eqref{eq:Ham} to support ballistic spin transport at non-zero temperature.
\subsection{Necessary and sufficient condition for perfectly ballistic spin transport}
\label{appendix:perfball}
Recall that if the spin current is exactly conserved, $\dot{\hat{J}}=i[\hat{H},\hat{J}]=0$, we expect perfectly ballistic spin transport and a non-zero spin Drude weight. By Eq. \eqref{eq:dJ/dt}, the spin current generated by the Hamiltonian Eq. \eqref{eq:Ham} is exactly conserved if and only if $D(u,v)$ is identically zero, i.e. the functional equation
\begin{equation}
\label{eq:purelyball}
(u-v) a(u) a(v) + (u+v) a(u+v) (b(u)-b(v)) = 0
\end{equation}
holds for all $u,v \in \mathbb{Z}$ such that $u,\,v,\, u+v \neq 0$. This functional equation is satisfied by both the nearest-neighbour XX and Haldane-Shastry models of interest in this work, which therefore exhibit perfectly ballistic spin transport. We now show that provided $b(r) \to 0$ as $r \to \infty$ (which is necessary on physical grounds) these are the only solutions to Eq. \eqref{eq:purelyball}, up to a choice of sublattice scale and a standard~\cite{takahashi2005thermodynamics} spin-flip-type unitary transformation on even-parity sites of each sublattice.

Specifically, we have the following:
\begin{theorem}
\label{thm:thm1}
Suppose that $a,\,b : \mathbb{Z}_{\neq 0} \to \mathbb{R}$ are even functions of their argument that solve Eq. \eqref{eq:purelyball}, where $a$ is non-zero and $b(r) \to 0$ as $r \to \infty$. Let $E_a = \{r \in \mathbb{Z}_{> 0} : a(r) \neq 0\}$ and $E_b = \{r \in \mathbb{Z}_{> 0} : b(r) \neq 0\}$ denote their respective supports. Then either
\begin{enumerate}
 \item $E_a = \{R\}$ for some integer $R > 0$ and $E_b = \emptyset$, which gives rise to $R$ identical nearest-neighbour XX models on the sublattices $\{R\mathbb{Z} + k : k=0,1,\ldots,R-1\}$ of $\mathbb{Z}$.
    \item $E_a = E_b = \{ n R : n \in \mathbb{Z}_{> 0}\}$ for some integer $R>0$ and either
    \begin{equation}
    \label{eq:usualHS}
    a(nR) = b(nR) = \frac{1}{n^2}a(R), \quad n \in \mathbb{Z}_{>0},
    \end{equation}
    which gives rise to $R$ identical Haldane-Shastry models on the sublattices $\{R\mathbb{Z} + k : k=0,1,\ldots,R-1\}$ of $\mathbb{Z}$, or
    \begin{equation}
    a(nR) = \frac{(-1)^{n+1}}{n^2}a(R), \quad b(nR) = -\frac{1}{n^2}a(R), \quad n \in \mathbb{Z}_{>0},
    \end{equation}
    which is unitarily equivalent to the latter case under conjugation by $\hat{U} = \prod_{k=0}^{R-1}\prod_{n \in \mathbb{Z}} 2\hat{S}^{z}_{2nR+k}$.
\end{enumerate}
\end{theorem}

\begin{proof}
First suppose that $b$ is the zero function. By assumption there exists $R > 0$ such that $a(R) \neq 0$. Then for all $x>0$, $x \neq R$, Eq. \eqref{eq:purelyball} implies that
\begin{equation}
(R-x)a(R)a(x) = 0,
\end{equation}
so that $a(x) = 0$. We deduce that $E_a = \{R\}$ and Case 1 of the theorem follows. 

Next suppose that $b$ is non-zero, so there exists some $x>0$ such that $b(x) \neq 0$. Suppose for a contradiction that the support of $a$ is bounded, i.e. $N_a = \max(E_a) < \infty$. Letting $u = x+N_a$, $v = -x$ in Eq. \eqref{eq:purelyball}, it follows that $x+N_a \notin E_a$ so that
\begin{equation}
N_a a(N_a)(b(x+N_a) - b(x)) = 0
\end{equation}
Thus $b(x+N_a)=b(x)$ and $x+N_a \in E_b$. Repeating this argument with $x$ replaced by $x+N_a$, it follows by induction that $b(x) = b(x+mN_a)$ for all $m \in \mathbb{Z}_{\geq 0}$, which contradicts our assumption on the decay of $b(r)$ as $r \to \infty$. Thus the support of $a$ is unbounded.

We next show that the support of $a$ forms an arithmetic progression. To see this, write $E_a = \{x_1,x_2,\ldots\}$ with $x_i > x_j$ whenever $i>j$.  Letting $i>j$ and substituting $u=x_i, \, v = \pm x_j$ into Eq. \eqref{appendix:perfball}, we deduce that
\begin{align}
(x_i-x_j)a(x_i) a(x_j) &= - (x_i+x_j)a(x_i+x_j) (b(x_i)-b(x_j)), \\
(x_i+x_j)a(x_i) a(x_j) &= - (x_i-x_j)a(x_i-x_j) (b(x_i)-b(x_j)).
\end{align}
Since the left-hand sides are non-zero, it must be the case that $a(x_i+x_j),\,a(x_i-x_j) \neq 0$ and $b(x_i) \neq b(x_j)$. In particular, $x_i \pm x_j \in E_a$ and $E_a$ is closed under sums and differences. Now let $R = x_1 = \min(E_a)$ and suppose there exists $y \in E_a$ such that $y \not\equiv 0 \mod R$. Then there exist integers $m > 0$ and $0 < k <R$ such that $y = mR + k$. By closure under differences, $y-mR = k \in E_a$ also, which contradicts minimality of $R$. We deduce that $x_j = jR$ and therefore $E_a = \{ n R: n \in \mathbb{Z}_{>0}\}$.

We now show that $E_b \subseteq E_a$. Suppose there exists some $u \in E_b \setminus E_a$. Then $u = mR + k$ for integers $m\geq 0$ and $ 0 < k < R$. Let $v = nR - k$ for some integer $n > m$. Then $u,v \notin E_a$ but $u+v = (m+n)R \in E_a$. Substituting into Eq. \eqref{eq:purelyball}, we deduce that $b(mR+k) = b(nR-k) \neq 0$ for all integers $n > m$, which contradicts the decay of $b(r)$. Thus $E_b \subseteq E_a$.

It remains to determine the specific behavior of $a(r)$ and $b(r)$ on $E_a$. To this end, let $u=nR$ and $v = \pm R$ in Eq. \eqref{eq:purelyball} for any integer $n>1$. Then
\begin{align}
b(R)-b(nR) =  \frac{(n-1)a(R)a(nR)}{(n+1)a((n+1)R)} = \frac{(n+1)a(R)a(nR)}{(n-1)a((n-1)R)},
\end{align}
implying that
\begin{equation}
a((n+1)R) = \frac{(n-1)^2}{(n+1)^2} a((n-1)R), \quad n=2,3,\ldots,
\end{equation}
from which it follows that
\begin{align}
\label{eq:casesfora}
a(nR) = \begin{cases}
    \frac{1}{n^2}a(R), & n \, \mathrm{odd} \\
    \frac{\lambda}{n^2}a(R), & n \, \mathrm{even}
\end{cases},
\end{align}
where it is useful to define $\lambda = 4 a(2R)/a(R)$. To proceed further, note that for any $u,\, v \in E_a$ with $u>v$ we can write Eq.~\eqref{eq:purelyball} as
\begin{equation}
b(u)-b(v) = - \frac{(u-v)a(u)a(v)}{(u+v)a(u+v)}.
\end{equation}
In particular, for any distinct positive integers $m$ and $n$, we have
\begin{align}
b((2m+1)R)-b((2n+1)R) &= \frac{1}{\lambda}\left(\frac{1}{(2m+1)^2} - \frac{1}{(2n+1)^2}\right)a(R),\\
b(2mR)-b((2n+1)R) &= \lambda \left(\frac{1}{(2m)^2} - \frac{1}{(2n+1)^2}\right)a(R), \\
b(2mR)-b(2nR) &= \lambda \left(\frac{1}{(2m)^2} - \frac{1}{(2n)^2}\right)a(R). 
\end{align}
Then, picking any $l > m > n$, the identity $b((2l+1)R)-b((2n+1)R) = (b((2l+1)R) - b(2mR)) + (b(2mR)-b((2n+1)R))$ demands that $\lambda^2=1$.

First consider the case $\lambda=1$. It follows by the above equations that $b(mR) - b(nR) = a(mR) - a(nR)$ for all positive integers $m$ and $n$. In particular,
\begin{equation}
b(nR) = a(nR) + (b(R)-a(R)), \quad n \in \mathbb{Z}_{>0},
\end{equation}
and our decay assumption on $b(r)$ sets $b(R) = a(R)$. We deduce that $E_b = E_a$ and 
\begin{equation}
a(nR) = b(nR) = \frac{1}{n^2}a(R), \quad n \in \mathbb{Z}_{>0}.
\end{equation}
Next consider the case $\lambda = -1$. In this case, Eq. \eqref{eq:casesfora} yields
\begin{equation}
a(nR) = \frac{(-1)^{n+1}}{n^2}a(R),\quad n \in \mathbb{Z}_{>0},
\end{equation}
while repeating the arguments above yields $b(nR) = -|a(nR)| = (-1)^n a(nR)$, which completes the proof of Claim 2.
\end{proof}

\subsection{Sufficient condition for a non-zero Mazur bound}
\label{appendix:twobody}
Even if the spin current operator $\hat{J}$ is not exactly conserved, it follows by the Mazur bound~\cite{mazur1969non,suzuki1971ergodicity} that ballistic transport can still occur if there exists another exactly conserved operator $\hat{K}$ that has sufficiently large overlap with $\hat{J}$ in the thermodynamic limit~\cite{Zotos_1997}. Indeed, this is precisely the scenario that Prosen demonstrated for the XXZ chain~\cite{ProsenBound} through the construction of suitable quasilocal $\hat{K}$. However, for a generic, chaotic, spin chain, we only expect local conservation laws $\hat{H}$ and $\hat{S}^z = \sum_{n\in \mathbb{Z}} \hat{S}_n^z$. In particular, if we assume evenness of the Hamiltonian $\hat{H}$ under the spatial inversion operator $\hat{\mathcal{I}}$ such that $\hat{\mathcal{I}}\hat{S}_{n}^z \hat{\mathcal{I}}^{\dagger} = \hat{S}_{-n}^z$ for all $n \in \mathbb{Z}$ and $\hat{\mathcal{I}}^\dagger = \hat{\mathcal{I}}$, it follows by unitarity of $\hat{\mathcal{I}}$ and cyclicity of the trace that for any thermal state at non-zero temperature, the thermal expectation values $\langle \hat{J}\hat{H} \rangle = \langle \hat{J}\hat{S}^z \rangle = 0$ by symmetry, since the conserved charges $\mathcal{\hat{I}}\hat{H}\hat{\mathcal{I}}^\dagger = \hat{H},\, \mathcal{\hat{I}}\hat{S}^z\hat{\mathcal{I}}^\dagger = \hat{S}^z$ are inversion-even but $\mathcal{\hat{I}}\hat{J}\hat{\mathcal{I}}^\dagger = -\hat{J}$ is inversion-odd. We deduce that generic, inversion-symmetric spin-chain Hamiltonians must exhibit diffusive spin transport at non-zero temperature, because there is no conserved operator that is symmetry-allowed to yield a non-zero Mazur bound. This argument generalizes straightforwardly to quantum or classical lattices of spins in any spatial dimension.

To justify the theoretical approach taken in the main text, we would ideally like to rule out the existence of such conserved operators rigorously. In this Appendix, we make partial progress along these lines and show that for the power-law interacting XXZ chains studied in this paper, there is no two-body conservation law $\hat{K}$ that could yield a non-zero Mazur bound for spin transport for model anisotropies $\Delta \neq 0, \pm 1$. While the restriction to two-body operators does not capture the phenomenology of e.g. the nearest-neighbour XXZ chain~\cite{ProsenBound}, it does capture the phenomenology of the Haldane-Shastry and nearest-neighbour XX models relevant to this work. 

The most general two-body operator $\hat{K}$ that respects the same discrete symmetries as $\hat{J}$ and generates a non-zero Mazur bound has the form
\begin{align}
    \hat{K} \coloneqq \sum_{n\in\mathbb{Z}, r\neq0} K(r) (\hat{S}_{n}^x \hat{S}_{n+r}^y - \hat{S}_{n}^y \hat{S}_{n+r}^x)
\end{align}
where $K$ is an odd function of its argument $K(-r)=-K(r)$. Requiring that $\hat{K}$ be a conserved operator leads to a nontrivial functional equation connecting $a(r)$, $b(r)$, and $K(r)$, namely
\begin{equation}
    K(u)a(v) - K(v)a(u) + K(u+v)(b(u)-b(v)) = 0
    \label{eq:two-body-operator}
\end{equation}
for all $u,v \in \mathbb{Z}$ such that $u,v,u+v \neq 0$, whose derivation mirrors that of Eq. \eqref{eq:dJ/dt}.
Eq. \eqref{eq:two-body-operator} can be viewed as an overdetermined system of linear equations for $K(v)$. We now show that this overdetermined system has no non-zero solutions for power-law-interacting XXZ models unless $|\Delta|$ is either zero or one.

To see this, pick an integer $x \neq 0$ for which $K(x)\neq 0$, and another integer $y \neq 0,\,x$ such that $z=x+y \neq 0$. Making the variable assignments $\{u=x,v=y\}$, $\{u=z,v=-x\}$, and $\{u=z, v=-y\}$ in Eq.~\ref{eq:two-body-operator} yields three simultaneous linear equations for $K(x)$, $K(y)$ and $K(z)$,
\begin{align}
    K(x)a(y)-K(y)a(x) + K(z)(b(x)-b(y)) &= 0, \\
    K(z)a(x)+K(x)a(z) + K(y)(b(z)-b(x)) &= 0, \\
    K(z)a(y)+K(y)a(z) + K(x)(b(z)-b(y)) &= 0.
\end{align}

Assuming that $a(x), \,a(y), \, b(x) - b(y) \neq 0$ subject to our assumptions on $x$ and $y$, we can solve for $K(z)$ to yield
\begin{equation}
    K(z) = \frac{K(y)a(x)-K(x)a(y)}{b(x)-b(y)} = \frac{-K(y)(b(z)-b(x))-K(x)a(z)}{a(x)} = \frac{K(x)(b(y)-b(z))-K(y)a(z)}{a(y)}.
\end{equation}

Combining the second and third expressions and the third and fourth expressions yields two simultaneous linear equations
\begin{align}
\label{eq:lin1}
\left(\frac{a(z)}{a(x)}- \frac{a(y)}{b(x)-b(y)}\right)K(x) + \left(\frac{a(x)}{b(x)-b(y)}+\frac{b(z)-b(x)}{a(x)}\right)K(y) &= 0,\\
\label{eq:lin2}
\left(\frac{a(z)}{za(x)}+\frac{b(y)-b(z)}{a(y)}\right)K(x) - \left(\frac{a(z)}{za(y)}-\frac{b(z)-b(x)}{a(x)}\right)K(y) &= 0.
\end{align}

Since $K(x) \neq 0$ by assumption, Eqs. \eqref{eq:lin1} and \eqref{eq:lin2} must have determinant zero,
 \begin{align}
    \left(\frac{a(z)}{a(x)}- \frac{a(y)}{b(x)-b(y)}\right)\left(\frac{a(z)}{a(y)}-\frac{b(z)-b(x)}{a(x)}\right)
    + \left(\frac{a(z)}{a(x)}+\frac{b(y)-b(z)}{a(y)}\right)\left(\frac{a(x)}{b(x)-b(y)}+\frac{b(z)-b(x)}{a(x)}\right) = 0.
\end{align}

Multiplying this equation by $a(x)a(y)(b(x)-b(y))$ finally yields the functional equation
\begin{equation}
\label{eq:conscond}
a(x)^2 (b(y)-b(z)) + a(y)^2(b(z)-b(x)) + a(z)^2 (b(x)-b(y)) + (b(x)-b(y))(b(y)-b(z))(b(z)-b(x)) = 0.
\end{equation}

Eq. \ref{eq:conscond} is a necessary condition for Eq. \ref{eq:two-body-operator} to have a non-zero solution. Imposing XXZ anisotropy $b(x) = \Delta a(x)$, Eq. \ref{eq:conscond} reduces to  
\begin{equation}
\Delta (1-\Delta^2) (a(x)-a(y))(a(y)-a(z))(a(z)-a(x)) = 0, \quad x,y \neq 0, \, x \neq y, \, z = x+y \neq 0.
\label{eq:necessary_cond}
\end{equation}

For the power-law XXZ models considered in this paper, for which $a(x) = 1/|x|^{\alpha}$, the steps leading to Eq. \eqref{eq:necessary_cond} are valid whenever $\Delta \neq 0$. In this specific case, we can take any $y > 0$, $y \neq x$ in Eq. \eqref{eq:necessary_cond} and obtain a contradiction whenever $|\Delta| \neq 1$. We deduce that for the power-law XXZ models defined by Eq. \eqref{eq:powerlawdeltapred}, there is no two-body conservation law that could generate a non-zero Mazur bound for spin transport away from the Haldane-Shastry and XX points.
\end{document}